\documentclass[11pt]{article}
\usepackage[utf8]{inputenc}
\usepackage{geometry}
\usepackage{epsfig}
\usepackage{boxedminipage}
\usepackage{hyperref}
\usepackage{csquotes}

\usepackage{amsthm}
\usepackage{amsmath,amsfonts,mathtools,amssymb}
\newtheorem{theorem}{Theorem}
\newtheorem{clm}{Claim}
\newtheorem{lemma}{Lemma}
\newtheorem{definition}{Definition}

\def\Remark{\par\noindent{\sl Remark\/}:\enspace}

\def\NP{\mbox{\bf NP}}
\def\P{\mbox{\bf P}}

\def\opt{\mbox{\tt OPT}}

\newcommand{\cut}[1]{\mbox{$(#1,\overline{#1})$}}
\newcommand{\abs}[1]{\mbox{$\left|#1\right|$}}

\newcommand{\set}[1]{\mbox{$\left\{#1\right\}$}}

\newcommand{\ceil}[1]{\mbox{$\lceil#1\rceil$}}
\newcommand{\remove}[1]{}

\newenvironment{lp}[2]{\[\begin{array}{rcll}
                        \mbox{#1} & & #2 \\ 
                        \mbox{subject to}}{\end{array}\]}
\newcommand{\cnstr}[4]{\\ #1 & #2 & #3 & #4}
 
\newcommand{\mcS}{\mathcal{S}} 
 
\newcommand{\mcB}{\mathcal{B}} 
\newcommand{\mcI}{\mathcal{I}} 
\newcommand{\mcC}{\mathcal{C}} 

\begin{document}
\title{Parallel Machine Scheduling to Minimize Energy Consumption}
\author{Antonios Antoniadis\footnote{Saarland University and Max-Planck-Institute for Informatics, Saarland University Campus, Saarbrücken, Germany. Supported by DFG grant AN 1262/1-1.} \and Naveen Garg\footnote{Indian Institute of Technology Delhi. Supported by a K.C. Iyer Chair} \and Gunjan Kumar\footnote{Tata Institute of Fundamental Research, Mumbai} \and Nikhil Kumar\footnote{Indian Institute of Technology Delhi} }
\maketitle
\begin{abstract}
Given $n$ jobs with release dates, deadlines and processing times we consider the problem of scheduling them on $m$ parallel machines so as to minimize the total energy consumed. Machines can enter a sleep state and they consume no energy in this state. Each machine requires $Q$ units of energy to awaken from the sleep state and in its active state the machine can process jobs and consumes a unit of energy per unit time. We allow for preemption and migration of jobs and provide the first constant approximation algorithm for this problem.
\end{abstract}
\thispagestyle{empty}
\clearpage
\setcounter{page}{1}
\section{Introduction}
Energy is an extremely important and scarce resource, and its consumption is progressively becoming a pivotal concern in modern societies. Computing environments account for a large fraction of the global energy consumption and alarmingly, this fraction is growing at a very high rate~\cite{greenpeace}. In response to this, modern hardware increasingly incorporates various energy-saving capabilities and scheduling algorithms need to be designed, not only for time and space considerations, but keeping energy consumption in mind as well.

We focus on one of the most common such power-management techniques called a \emph{power-down mechanism}, which refers to the ability of the processor to transition into a sleep state where it consumes negligible energy. Since \enquote{waking-up} the processor requires a certain amount of energy, there is a trade-off to be had between the energy saved by residing in the sleep state and the energy expended in transitioning back to the active state. Intuitively, one should aim to keep the number of transitions to the sleep states low and once in a sleep state remain in it for as long as possible.

Consider a set of jobs with individual release times, deadlines and processing times, that are to be processed on either a single or a multiprocessor system equipped with a powerdown mechanism. The processor consumes one unit of energy per unit of time when in the active state and no energy when in the sleep state. Transitioning from the sleep state to the active state incurs a fixed energy cost. Preemption and migration of jobs is allowed but no job can be simultaneously processed on more than one machine. The goal is to produce a feasible schedule which consumes the minimum energy (or report that no feasible schedule exists). In Graham's notation, and with $E$ being the appropriate energy function the problems we study can be denoted as $1|r_j;\overline{d}_j;\text{pmtn}|E$ and $m|r_j;\overline{d}_j;\text{pmtn}|E$ respectively.

The problem on a single machine was first stated in~\cite{iraniSGtalg07}, where a greedy $2$-approximation algorithm  called \emph{Left-To-Right} was presented. Roughly speaking, Left-To-Right tries to keep the machine at its current state (active or asleep) for as long as possible. However the computational complexity of the problem remained open and was repeatedly posed as an important open question, in particular because \enquote{many seemingly more complicated problems in this area can be essentially reduced to this problem} (c.f.~\cite{IraniP05}). The complexity question, for the single-machine setting, was eventually settled, initially by Baptiste~\cite{baptiste2006scheduling}, who gave a $O(n^7)$-time algorithm for the case of unit-size jobs and subsequently by Baptiste et. al.~\cite{baptiste2007polynomial} who achieved a running time of $O(n^4)$ for unit sized jobs and $O(n^5)$  when jobs can have arbitrary processing times. Both algorithms are based on a rather involved dynamic programming approach. 

The multiprocessor case turns out to be much more challenging than the single processor one, and obtaining any algorithm for it with a non-trivial performance guarantee has been a major open problem~\cite{baptiste2007polynomial}. It is also an open problem whether the problem is \NP-hard. The difficulty in obtaining a good approximation algorithm seems to arise from two aspects: First, it is not clear how to design a dynamic programming table of polynomial size when the jobs have arbitrary sizes, and a job is not allowed to run parallel to itself. Secondly, structural properties of an optimal schedule can be locally extracted in a single machine environment in contrast to the multi-machine case. As an example, we know that a single machine will be active for at least one time-point within the interval between the release time and the deadline of every job, but the number of active machines at such a time-point in the multiprocessor setting could range from just one to all available machines. As a result, there has been only one previous result with for the multiprocessor setting;  by Demaine et al.~\cite{demaine2007scheduling} who extended the dynamic program of Baptiste~\cite{baptiste2006scheduling} and showed an $O(n^7m^5)$-time algorithm for the special case of unit-size jobs and $m\ge 1$ machines.

\subsection{Our Contribution}

In Section~\ref{sec:single} we present a pseudo-polynomial time algorithm for single machines that produces a feasible schedule of total energy at most $\opt+P$ where \opt\ is the minimum energy of any fractional solution and $P$ the sum of processing times. The algorithm is based on an elegant linear programming relaxation which we  extend to the multiprocessor case in a later section. We show that the solution of the linear program relaxation can be decomposed into a convex combination of integer solutions. Since the relaxation has a strictly positive integrality gap, none of the integer solutions in the decomposition may be feasible. We overcome this by showing how an (infeasible) integer solution can be extended into a feasible solution while increasing the total energy consumption by only an additive $P$. Note that $P$ is also a lower-bound on the optimal energy consumption and hence our algorithm can be viewed as a 2-approximation. Let $n$ be the number of jobs and $D$ the maximum deadline. We prove the following theorem in Section~\ref{sec:single}.
\begin{theorem}
There is an algorithm with running time polynomial in $n,D$ for single machines that produces a schedule of total energy at most $\opt + P$.
\end{theorem}

Building upon ideas for the single machine case, we develop, in Section~\ref{sec:multiple} the first constant-factor approximation algorithm for the multiple machines case. Checking the feasibility  of an instance and formulating a linear program to minimize energy is much more involved in the setting of multiple machines. The intervals comprising the integer solutions in the convex decomposition of the optimum fractional solution  are not disjoint anymore, and extending the intervals appropriately in order to obtain feasibility is much more challenging now. We overcome these obstacles and present a pseudo-polynomial time algorithm that produces a feasible schedule of total energy at most $2\opt+P$. We prove the following theorem in Section~\ref{sec:multiple}
\begin{theorem}
There is an algorithm with running time polynomial in $n,D$ for $m$ parallel machines that produces a schedule of total energy at most $2\opt + P$.
\end{theorem}

Finally, in the Appendix, we show that the running time of our algorithms can be made polynomial in $n, 1/\epsilon$; we incur a $(1+\epsilon)$ loss in the approximation factor in this process.

\subsection{Further Related work}

An important generalization of our problem would be speed scaling with a sleep state, where the processor can vary its speed when in the active state in order to further save energy. The power consumption of the processor when it is active depends on its speed. In a processor with only speed scaling (and no sleep state) one tries to keep the processor speed as low as possible (since power is a convex function of speed). However with both speed scaling and a sleep state it is often beneficial to run the processor at faster speeds in order to increase the length of the subsequent sleep states, a technique commonly referred to as \emph{race to idle}.  
Speed scaling with a sleep state was first introduced in~\cite{iraniSGtalg07} who gave a $2$-approximation algorithm for the problem. This result was later improved to a $4/3$-approximation by Albers and Antoniadis~\cite{AlbersA14racetoidle1}, and eventually to a fully polynomial time approximation scheme (FPTAS) by Antoniadis et al.~\cite{AntoniadisHO15racetoidle2}. This is the best result one can hope for (unless $\P=\NP$), as the problem is known to be \NP-hard~\cite{AlbersA14racetoidle1,KumarS15a}.

Another problem similar to ours is that of minimizing the number of gaps (a gap is a contiguous interval during which the processor is idle) in the schedule. If one is interested in exact solutions then this is a special case of our problem since by choosing a large value for energy consumed in the active state we can ensure that every idle period results in a transition to the sleep state; thus the optimal schedule also minimizes the number of gaps. Chrobak et al.~\cite{chrobak2017greedy} gave a simple $2$-approximation algorithm for the gap minimization problem with a running time of $O(n^2 \log n)$ and memory just $O(n)$. Demaine et al.~\cite{demaine2007scheduling} gave an exact algorithm for the multiprocessor gap minimization problem with unit-size tasks. Several further generalizations - for example the set-cover-hard case when each job has several disjoint release time-deadline intervals to choose from - of the problem were considered in~\cite{demaine2007scheduling,DemaineZ10}.

Finally, one may consider the setting where one knows exactly when the processor (or how many processors at each point in time) need to be active in order to execute jobs, and has to decide about when to transition the processor(s) between the states. Although the offline version of the problem with a single processor equipped just with one active and one sleep state becomes trivial, the online version turns out to be a generalization of the well-known ski-rental problem. Additionally considering processor(s) with sleep states of various depths (each having an individual power consumption and an individual cost for transitioning back to the active state) leads to many interesting algorithmic problems both in the offline and in the online scenarios that have been studied by Albers~\cite{Albers17spaa}, Augustine et al.~\cite{augustineIS08}, as well as  Irani et al.~\cite{IraniSG03multipowerdown}.

\section{Preliminaries}
We are given a set of jobs \set{j_1,j_2,\ldots, j_n}; job $j_i$ has release time $r_i$, deadline $d_i$ and processing time $p_i$ and we assume that all these quantities are non-negative integers. 
Let $r_{min}$ and $d_{max}$ be the earliest release time and furthest deadline of any job; it is no loss of generality to assume $r_{min} = 0$ and $d_{max}=D$. For $t\in \mathbb{Z}^+$, let $[t,t+1]$ denote the $t^{\rm th}$ {\em time-slot}. Let $I=[t,t'], t, t'\in \mathbb{Z}^+, t<t'$ be an {\em interval}. The length of $I$, denoted by \abs{I} is $t'-t$. We use $t\in I_1$ to denote $a_1\le t\le b_1$. 

Two intervals $I_1=[a_1,b_1]$ and $I_2=[a_2,b_2]$ {\em overlap} if there is a $t$ such that $t\in I_1$ and $t\in I_2$. Thus two intervals which are right next to each other would also be considered overlapping. Intervals which do not overlap are considered {\em disjoint}. $I_1$ is {\em contained} in $I_2$, denoted $I_1\subseteq I_2$, if $a_1\le a_2 < b_2\le b_1$ and it is {\em strictly contained} in $I_2$, denoted $I_1\subset I_2$, if $a_1 < a_2 < b_2 < b_1$. 

At any time-slot, a machine can be in the \emph{active} or the \emph{sleep} state. For each time-slot that a machine is in the active state, one unit of power is required whereas no power is consumed in the sleep state. However, $Q$ units of energy (called \emph{wake up} energy) are expended when the machine transitions  from the sleep to the active state. In its active state, the machine can  either process a job (in which case we refer to it as being \emph{busy}) or just be \emph{idle}. On the other hand the machine cannot perform any processing while in the sleep state. Note that if a machine is not required to do any processing for $L$ consecutive time-slots, then it is advantageous to transition it to the sleep state when $L>Q$ whereas for $L\le Q$ it is preferable to keep it active but idle.

A machine can process at most one job in any time-slot and a job cannot be processed on more than one machine in a time-slot. However, job preemption and migration are allowed, i.e., processing of a job can be stopped at any time and resumed later on the same or on a different machine. A job $j_i$ must be processed for $p_i$ time-slots in $[r_i,d_i]$. Any assignment of jobs to machines and time slots satisfying the above conditions is called a (feasible) \emph{schedule}. We assume that the machine is initially in the sleep state. Therefore, the energy consumed by a schedule is the total length of the intervals during which the machine is active plus $Q$ times the number of intervals in which the machine is active. The objective of the problem is to find a schedule which consumes minimum energy.

\section{An additive \texorpdfstring{$P$}{P} approximation for single machines}
\label{sec:single}
We first show how to schedule jobs on a single machine so that the total energy consumption is at most $P$ more than the optimum. For any $[a,b]\subseteq[0,D]$ (recall $D$ is the furthest deadline of any job), let $V(a,b)=\sum_{i:[r_i,d_i]\subseteq[a,b]} p_i$ be the total processing time of jobs whose release and deadline are within $[a,b]$. For an instance to be feasible it is necessary that for all $0\le a < b\le D$, $V(a,b)\le b-a$. The Earliest Deadline First (EDF) algorithm for scheduling jobs with release dates and deadlines can also be used to establish the sufficiency of this condition. 
 
 Motivated by this necessary and sufficient condition for determining if an instance is feasible, we consider the following Integer Program for minimizing total energy consumed. For $I\subseteq [0,D]$ let $x_I$ be a variable which is 1 if the machine becomes active at the start of $I$ and remains so till its end when it transitions back to the sleep state; $x_I$ is 0 otherwise. Since the machine uses $Q$ units of energy to wake-up at the start of $I$ and \abs{I} units to run during this interval, the objective is to minimize $\sum_I x_I(Q+\abs{I})$. We next discuss the constraints of this IP.
 \begin{enumerate}
     \item The intervals in which the machine is active are disjoint and hence for $0\le t\le D$, $\sum_{I: t\in I} x_I\le 1$.
     \item To ensure that jobs can meet release dates and deadlines when scheduled within active intervals we add the constraint that for all $0\le a < b\le D$, $V(a,b)\le \sum_I x_I\abs{I\cap[a,b]}$.
     \item For any job $j_i$, the machine should be active at some point during $[r_i,d_i]$. Hence \\ $\sum_{I:I\cap[r_i,d_i] \neq \varnothing} x_I \geq 1$
 \end{enumerate}
 This gives us the following integer program.
 \begin{lp}{minimize}{\sum_I x_I(Q+\abs{I})}
 \cnstr{\sum_{I: [t,t+1]\in I} x_I}{\le}{1}{0\le t\le D-1}
 \cnstr{\sum_I x_I\abs{I\cap[a,b]}}{\ge}{V(a,b)}{0\le a < b\le D}
 \cnstr{\sum_{I:I\cap[r_i,d_i]\neq\varnothing} x_I}{\geq}{1}{1\le i\le n}
 \cnstr{x_I}{\in}{\set{0,1}}{I\subseteq[0,D]}
 \end{lp}

Consider a feasible solution to this IP and let $\mcI=\set{I|x_I=1}$. A time-slot $[t,t+1]$ is {\em active} if it is contained in some interval of $\mcI$.
\begin{clm}
 Every job $j_i$ can be assigned to $p_i$ active time slots in $[r_i,d_i]$ such that each active time-slot is assigned to at most 1 job.
\end{clm}
\begin{proof}
Construct a bipartite graph $G=(U,V,E)$. For every job $j_i$ we have $p_i$ vertices in $U$ and for every active time slot we have a vertex in $V$. $E$ has an edge between a vertex corresponding to job $j_i$ and a vertex corresponding to the active time-slot $[t,t+1]$ iff $[t,t+1]\subseteq [r_i,d_i]$. We want to find a matching in $G$ which matches all vertices of $U$.

For contradiction assume that there is no such matching. By Hall's theorem there exists a Hall set $S\subseteq U$ such that $\abs{\Gamma(S)}<\abs{S}$ where $\Gamma(S)$ are the vertices in $V$ adjacent to vertices in $S$. Let $S$ be a minimal Hall set. Two vertices in $U$ corresponding to the same job have identical neighbors in $V$ and hence it is no loss of generality to assume that $S$ contains all vertices corresponding to the same job. This allows us to view $S$ as a set of jobs; \abs{S} then equals the total processing time of the jobs in $S$.

Consider the union of intervals $[r_i,d_i]$ where $j_i$ is a job in $S$. The minimality of $S$ implies that this union is a single interval, say $[a,b]$. Note that $V(a,b)\ge\abs{S}$ and $\abs{\Gamma(S)}$ is the number of active time slots in $[a,b]$. From the second set of constraints of the IP it follows that $\abs{S}\le V(a,b)\le\abs{\Gamma(S)}$ which contradicts our assumption that $S$ is a Hall set.
\end{proof}
The above claim implies that an optimum solution to the integer program gives a feasible schedule which minimizes energy. We relax the integrality constraint on $x_I$ to $0\le x_I\le 1$ and solve the resulting linear program. Let $x$ be the optimum fractional solution and let $\mcI=\set{I|x_I>0}$. We will next show that $x$ be decomposed into a convex combination of integer solutions.

{\bf Ordering intervals in $\mcI$}: Let $[a,d],[b,c]\in\mcI$, $[b,c]\subset[a,d]$ and $x_{[a,d]}=x_{[b,c]}=\alpha$. we replace these intervals in $\mcI$ with intervals $[a,c],[b,d]$ and set $x_{[a,c]}=x_{[b,d]}=\alpha$. Doing so does not make $x$ infeasible nor does it change the objective value. If $\beta=x_{[a,d]}>x_{[b,c]}=\alpha$ then we replace these intervals in $\mcI$ with three intervals $[a,d],[a,c],[b,d]$ and set $x_{[a,d]}=\beta-\alpha$ and $x_{[a,c]}=x_{[b,d]}=\alpha$. The case when $\beta=x_{[a,d]}<x_{[b,c]}=\alpha$ is handled similarly. We repeat this process whenever an interval in $\mcI$ strictly contains another interval in $\mcI$. Finally, order the intervals in $\mcI$ by their start-times; intervals which have the same start-time are ordered by their end-times. Let $\prec$ denote this total order on intervals of $\mcI$. Note that since no interval is strictly contained in another, we would get the same ordering if intervals were ordered by their end-times with intervals having the same end-time ordered by their start-times. 

{\bf Decomposing $x$ into a convex combination of integer solutions}: For $I\in\mcI$ let $s_I$ be the fractional part of $\sum_{I'\prec I} x_{I'}$; thus $0\le s_I<1$. For $k$, $0\le k<1$ construct $\mcI_k\subseteq \mcI$ as follows: $I\in \mcI_k$ iff either $s_I\le k < s_I+x_I$ or $s_I\le k+1 < s_I+x_I$.
\begin{clm}\label{cl:disjoint}
 The intervals in $\mcI_k$ are disjoint.
\end{clm}
\begin{proof}
Let $I_1,I_2\in\mcI_k$, $I_1\prec I_2$ and $I_1\cap I_2\neq\varnothing$. Since $I_1,I_2\in\mcI_k$ and $I_1\prec I_2$, we get $\sum_{I_1\preceq I\preceq I_2} x_I > 1$. Since $I_1,I_2$ are not disjoint, all intervals $I$ such that $I_1\preceq I\preceq I_2$ have a common overlap, say at time $t$. But this violates the LP-constraint $\sum_{I:t\in I} x_I\leq 1$ and yields a contradiction.
\end{proof}
Let $0=s_1<s_2<\cdots<s_m< 1$ be the distinct values in the set \set{s_I, I\in\mcI}; note that $m\le\abs{\mcI}$. From our construction of $\mcI_k$ it follows that for all $k\in[s_j,s_{j+1})$ the set $\mcI_k$ are identical; let $\mcC_j$ denote this set and we assign it a weight $w_j=s_{j+1}-s_j$ (or $1-s_m$ for the border case when $j=m$). 
By Claim~\ref{cl:disjoint}, each ``solution" $\mcC_j, 1\le j\le m$ is a set of disjoint intervals.
\begin{clm}
 The solutions $\mcC_j$ and weights $w_j$, $1\le j\le m$, form a convex decomposition of the fractional solution $x$.
\end{clm}
\begin{proof}
First note that for all $1\le j\le m$, $w_j\ge 0$ and $\sum_{j=1}^m w_j=1$. Now consider an interval $I\in\mcI$ and let $s_I=s_a$ and $s_I+x_I=s_b$, $b>a$. The interval $I$ appears in solutions $\mcC_a,\mcC_{a+1},\ldots,\mcC_{b-1}$ and these have a total weight $s_b-s_a=x_I$.
\end{proof}
\Remark{An alternate procedure to construct this convex decomposition of $x$ would be to replace each interval $I\in\mcI$ with $x_I/\epsilon$ intervals where $\epsilon$ is such that $x_I/\epsilon$ is an integer for all $I\in\mcI$. Let $\mcI'$ be the multiset of intervals obtained. Consider intervals in $\mcI'$ in the order $\prec$ and assign them to solutions $\mcC_1,\mcC_2,\ldots,\mcC_{1/\epsilon}$ in a round robin manner. Although easy to present, this procedure has the disadvantage that the number of solutions in the convex decomposition is $1/\epsilon$ and $\epsilon$ which is the granularity of the fractional solution $x$, could be exponentially small. One could round $x$ to multiples of $\epsilon$ for a suitable choice of $\epsilon$ but this would then incur a multiplicative constant in the approximation guarantee. The procedure presented above is conceptually similar to this round-robin assignment.}

{\bf Extending Intervals}: Although $\mcC_j, 1\le j\le m$ is a set of disjoint intervals it need not be a feasible solution, i.e. it could be that jobs cannot meet release dates and deadlines if they have to be scheduled within intervals of $\mcC_j$. This is illustrated by the example in Figure~\ref{fig:bad-example}, the details of which can be found in the Appendix.
\begin{figure}
    \centering
    \includegraphics[scale=0.6]{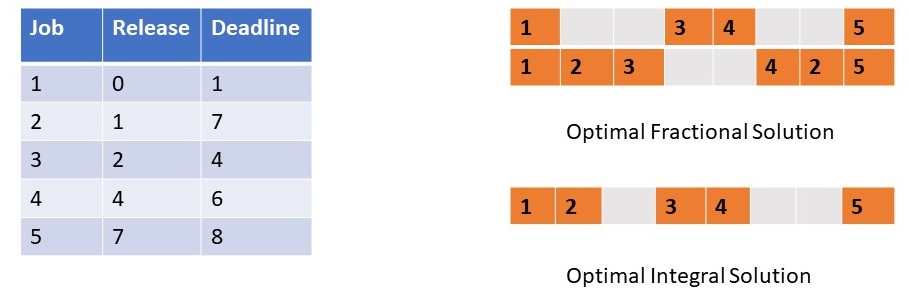}
    \caption{An instance where solutions in the convex decomposition are not all feasible. All tasks are unit size. The top right shows the two solutions $\mcC_1,\mcC_2$ in the convex decomposition of the optimum fractional solution. The total length of intervals in $\mcC_1$ is 4 which is less than the total processing time of jobs and implies $\mcC_1$ is infeasible.}
    \label{fig:bad-example}
\end{figure}
We next show that we can extend the intervals in any solution $\mcC_j, 1\le j\le m$ by at most $P$ units to get a feasible solution, $\mcC'_j$. 
\begin{lemma}
Let $\mcC=\mcC_j, 1\le j\le m$ be a solution from the convex decomposition of $x$. $\mcC$ can be converted into a feasible solution $\mcC'$ by increasing the total length of intervals  in $\mcC$ by at most $P$.
\end{lemma}
\begin{proof}
A slot $[t,t+1]$ is {\em active} if it is contained in some interval in $\mcC$. Let $s(a,b)$ be the number of active slots in the interval $[a,b]\subseteq[0,D]$ and $\delta(a,b)=\max(0,V(a,b)-s(a,b))$ its {\em deficiency}.

If $\mcC$ is infeasible there exists $[a,b]$ such that $\delta(a,b)>0$. Among all intervals with positive deficiency consider those whose end-time is the least and let these be $[a_1,t],[a_2,t],\ldots, [a_k,t]$ where $t>a_1>a_2>\cdots>a_k$. Let $P_t$ be the total processing time of jobs whose deadline is $t$. For $1\le i\le k$, $V(a_i,t)\le V(a_i,t-1)+P_t$ and since $\delta(a_i,t-1)=0$ we have $\delta(a_i,t)\le P_t$. 

We now show how to extend intervals in $\mcC$ by $P_t$ time-slots so that deficiency of intervals $[a_i,t], 1\le i\le k$ becomes 0. 
\begin{clm}
$\mcC$ contains an interval which overlaps $[a_1,t]$.
\end{clm}
\begin{proof}
$\delta(a_1,t)>0$ implies $V(a_1,t)>0$ which in turn implies that there exists a job $j_i$ such that $[r_i,d_i]\subseteq[a_1,t]$. The third set of constraints of the integer program ensure that the sum of $x_I$ where $I\in\mcI$ and $I\cap[r_i,d_i]\neq\phi$ is at least 1. By our procedure for building the convex decomposition it follows that at least one of these intervals is in $\mcC$. Since this interval overlaps $[r_i,d_i]$ it also overlaps $[a_1,t]$ proving the claim.
\end{proof}
Let $I\in\mcC$ overlap $[a_1,t]$. We first extend $I$ to the right till we have included time-slot $[t-1,t]$ and continue by extending $I$ to the left, perhaps combining with other intervals of $\mcC$ in this process. We stop when $P_t$ time-slots have been added or when all time-slots before $t$ have been included. Consider the interval $[a_i,t]$. Either we have added $P_t$ time slots in this interval or extended $I$ to include all time-slots in this interval. In the former case the deficiency of $[a_i,t]$ is reduced to 0. In the later case $s(a_i,t)=t-a_i\ge V(a_i,t)$, where the second inequality follows from the fact that the instance is feasible. Hence $\delta(a_i,t)=0$.

After having reduced to zero the deficiency of all intervals ending at $t$, we find the next set of intervals with positive deficiency whose end-time is the least. The process continues till all intervals have zero deficiency. Note that the intervals of $\mcC$ are extended by at most $\sum_t P_t=P$ time-slots. 

\end{proof}
Since the number of intervals in $\mcC'_j$ equals the number of intervals in $\mcC_j$ and the total length of intervals in $\mcC'_j$ exceeds the total length of intervals in $\mcC_j$ by at most $P$, the energy consumed by the solution $\mcC'_j$ is at most $P$ more than the energy consumed by $\mcC_j$. Since this is true for all solutions $\mcC'_j, 1\le j\le m$, the solution of minimum cost among these has cost at most $P$ more than the optimum fractional solution. 
\begin{theorem}
Given $n$ jobs with release dates, processing times and deadlines in $[0,D]$, there is an algorithm with running time polynomial in $n,D$ which schedules these jobs on a single machine such that the total energy consumption is at most $\opt +P$ where $P$ is the sum of processing times.
\end{theorem}

\section{Deadline Scheduling on Parallel Machines} \label{feasible}
\newcommand{\dsoi}{{\tt deadline-scheduling-\linebreak[2]on-intervals}}
In this section we prove a necessary and sufficient condition for scheduling jobs on $m$ parallel machines so that all release dates and deadlines are met. While this is a standard problem in an undergraduate Algorithms course we repeat the argument here since it will be useful in developing the linear program for minimizing energy consumption in the next section. 

Recall we are given $n$ jobs. Job $j_i, 1\le i\le n$ requires $p_i$ units of processing, is released at time $r_{i}$ and has deadline $d_i$. The jobs are to be scheduled on $m$ identical machines and we allow for preemption and migration. An instance is feasible iff for every job $j_i, 1\le i\le n$ we can assign $p_i$ distinct time-slots during $[r_i,d_i]$ such that no time-slot is assigned to more than $m$ jobs. 

For reasons that will become clear later, we consider a minor generalization of the above problem which we refer to as \dsoi. Instead of $m$ machines, we are given $k$ supply-intervals, $\mcI=\set{I_{1},I_{2}, \ldots, I_{k}}$ and are required to schedule the given jobs within these intervals. Let $s_j,t_j$ denote the start and end-times of interval $I_j$. The intervals in $\mcI$ need not be disjoint; however any point in time is contained in at most $m$ intervals. Note that if each interval in $\mcI$ was $[0,D]$ then we would recover the problem of scheduling on parallel machines. An instance of this problem is thus specified by the processing time, release date and deadline of each of the $n$ jobs and the start and end-times of the $k$ supply-intervals. The feasibility of an instance can be checked by formulating it as a problem of finding a flow in a suitable network. 

\begin{figure}[ht]
\centering
\includegraphics[width=4.5 in]{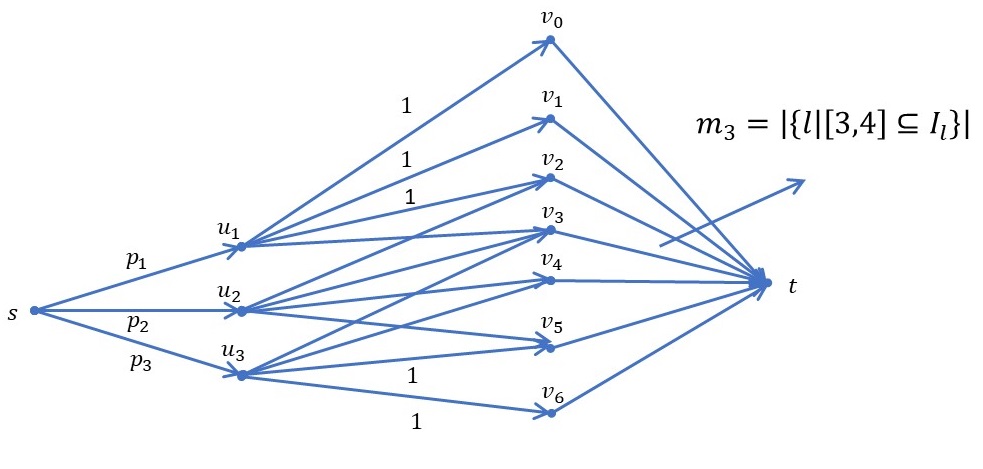}
\caption{Network $G=(V,E)$ for checking feasibility of an instance}
\label{network}
\end{figure}

Construct a network $G=(V,E)$ with source $s$, sink $t$, a vertex $u_i$ for each job $j_i$ and a vertex $v_t$ for each time-slot $[t,t+1], 0\le t\le D-1$. Vertex $u_i$ has edges to vertices \set{v_t | [t,t+1]\subseteq[r_i,d_i]} of capacity 1 and an edge from $s$ of capacity $p_i$. Let $m_t$ be the number of intervals in $\mcI$ which contain the time-slot $[t,t+1]$. Vertex $v_t$ has an edge to the sink $t$ of capacity $m_t$. Let $c:E\rightarrow \mathbb{R}^{+}$ denote the capacity function on the edges. 

The $s$-$t$ cut $(\set{s},V-\set{s})$ has capacity $P=\sum_{j=1}^n p_j$ and so the maximum flow between $s$ and $t$ cannot exceed $P$.
\begin{lemma}\label{le:fesible}
An instance of \dsoi\ is feasible iff $P$ units of flow can be sent from $s$ to $t$ in the network $G$ with capacities given by $c$.
\end{lemma}
\begin{proof}
Let $f:E\rightarrow\mathbb{Z}_{\ge 0}$ be an $s$-$t$ flow of value $P$. Since edge capacities are integral $f$ can also be assumed to be integral. We use $f$ to determine an assignment of jobs to time-slots. If $f(u_i,v_t)=1$ then we assign job $j_i$ to the time-slot $[t,t+1]$. Since $f(v_t,t)\le m_t$ the number of jobs assigned to time-slot $[t,t+1]$ cannot exceed the number of intervals in $\mcI$ containing this time-slot. Since $f$ has value $P$ which is the capacity of the cut $(\set{s},V-\set{s})$, all edges incident to $s$ are saturated. Hence $f(s,u_i)=p_i$ which implies that job $j_i$ is assigned to exactly $p_i$ time-slots in $[r_i,d_i]$. This assignment of jobs to time-slots is therefore a feasible schedule. 

For the converse, consider a schedule, $\mcS$, which respects release dates and deadlines. We build a flow $f$ from $s$ to $t$ of value $P$. If job $j_i$ is processed in time-slot $[t,t+1]$ in $\mcS$ then $f(u_i,v_t)=1$; since $[t,t+1]\subseteq[r_i,d_i]$, the edge $(u_i,v_t)$ is in $E$ and has capacity 1. The flow on edges entering $t$ and leaving $s$ is determined by conservation. Note that at most $m_t$ jobs could be scheduled in the time-slot $[t,t+1]$ and hence the flow on edge $(v_t,t)$ does not exceed its capacity. Since in schedule $\mcS$, job $j_i$ is processed for $p_i$ units, the flow on edge $(s,u_i)$ equals $p_i$ which implies that the total flow from $s$ to $t$ is $P$.
\end{proof}

Let \cut{S} be an $s$-$t$ cut and $c(S)$ denote its capacity.
\begin{clm}\label{cl:basic}
If $c(S)<P$ then $S\cap \set{v_0,v_1,\ldots,v_{D-1}}\neq\varnothing$.
\end{clm}
\begin{proof}
For contradiction assume that $S$ does not contain any vertex from the set \set{v_0,v_1,\ldots,v_{D-1}}. Then the capacity of the cut \cut{S} is $\sum_{i:u_i\in S} (d_i-r_i) + \sum_{i:u_i\not\in S} p_i$. If for job $j_i$, $d_i-r_i < p_i$ then the instance is trivially infeasible. Hence we assume that $d_i-r_i\geq p_i, 1\le i\le n$, and this implies that the capacity of the cut \cut{S} is at least $\sum_{i=1}^n p_i=P$. 
\end{proof}

We aggregate the time-slots corresponding to vertices in $S\cap \set{v_0,v_1,\ldots,v_{D-1}}$ into a minimal set of intervals, $Q(S)$. No two intervals in $Q(S)$ are overlapping since we could combine them and obtain a smaller set of intervals. Recall that if two intervals share an end-point then we consider them overlapping.  
\newcommand{\fv}[2]{\mbox{\tt fv($#1,#2$)}}
\newcommand{\defi}[1]{\mbox{\tt def($#1$)}}

\begin{definition}
The {\em forced volume} of a job $j_i$ with respect to an interval $[a,b]$, denoted by \fv{j_i}{[a,b]}, is the minimum volume of $j_i$ that must be processed during $[a,b]$ in any feasible schedule. Let $Q$ be a set of disjoint intervals. The {\em forced volume} of job $j_i$ with respect to $Q$ denoted by \fv{j_i}{Q}, is the minimum volume of $j_i$ that must be processed during the intervals in $Q$ in any feasible schedule.
\end{definition}
If $I_1, I_2$ are disjoint intervals then $\fv{j_i}{I_1}+\fv{j_i}{I_2}\leq \fv{j_i}{I_1\cup I_2}$. For instance suppose $I_1=[0,3]$, $I_2=[5,8]$, $r_1=2$, $d_1=6$ and $p_1=3$. Then \fv{j_1}{I_1}=\fv{j_1}{I_2}=0 but \fv{j_1}{I_1\cup I_2}=1. Note that the forced volume of a job $j_i$ with respect to an interval $[a,b]$ is independent of the supply-intervals and depends only $a,b,p_i,r_i$ and $d_i$. For instance, if $r_i<a<d_i<b$ then $\fv{j_i}{[a,b]}=\max(0,r_i+p_i-a)$. Similarly, if $r_i\le a<b \le d_i$ then $\fv{j_i}{[a,b]}=\max(0,p_i-(a-r_i)-(d_i-b))$.

\begin{definition}
Let $Q$ be a set of disjoint intervals. The deficiency of $Q$, denoted by \defi{Q}, is the non-negative difference between the sum of the forced volume of all jobs with respect to $Q$ and the total volume of jobs that can be processed in $Q$. Thus $$\defi{Q}=\max \left( 0,\sum_{i=1}^n \fv{j_i}{Q}-\sum_{t:[t,t+1]\subseteq Q} m_t \right).$$
\end{definition}
Note that deficiency of $Q$ also depends on the supply intervals in the instance. From the above definition it follows that if a set of disjoint intervals, $Q$, has positive deficiency then the instance is infeasible. The following lemma will help us argue the converse. 
\begin{lemma}\label{le:deficiency}
Let \cut{S} be a $s$-$t$ cut in $G$. Then $\defi{Q(S)}+c(S)\geq P$. The inequality holds with an equality if \cut{S} is a minimum $s$-$t$ cut. 
\end{lemma}
\begin{proof}
We consider each vertex in $S$ and count the total capacity of edges in the cut \cut{S} incident to this vertex. \begin{enumerate}
    \item For the source $s$, this quantity is $\sum_{i:u_i\not\in S} p_i$.
    \item Let $u_i\in S$ and $c_i$ be the number of edges from $u_i$ to vertices in $\overline{S}$. If $c_i>p_i$  then $\fv{j_i}{Q(S)}=0$ and if $c_i\le p_i$ then $\fv{j_i}{Q(S)}=p_i-c_i$. Hence $c_i\geq p_i - \fv{j_i}{Q(S)}$.
    \item If $v_t\in S$ then the edge $(v_t,t)$ of capacity $m_t$ is in \cut{S}.
\end{enumerate}
Combining these we get 
\begin{eqnarray*}
c(S) & \ge & \sum_{i:u_i\not\in S} p_i + \sum_{i:u_i\in S} (p_i-\fv{j_i}{Q(S)})+ \sum_{t:[t,t+1]\subseteq Q(S)} m_t\\
& \ge & \sum_{i=1}^n p_i - \sum_{i=1}^n \fv{j_i}{Q(S)})+ \sum_{t:[t,t+1]\subseteq Q(S)} m_t\\
& \ge & P - \defi{Q(S)}
\end{eqnarray*}
which proves the first part of the lemma.

Let \cut{S} be a minimum $s$-$t$ cut. 
\begin{enumerate}
    \item If $u_i\in S$ then $c_i\le p_i$ or else we would have moved $u_i$ to $\overline{S}$ to obtain a cut of smaller capacity. Hence $c_i=p_i-\fv{j_i}{Q(S)}$.
    \item In a maximum $s$-$t$ flow, flow on edge $(u_i,v_t)$, $u_i\not\in S, v_t\in S$, is 0. Since $p_i$ units enter $u_i$, this implies that $\fv{j_i}{Q(S)}=0$ and hence $\sum_{i:u_i\in S} (p_i-\fv{j_i}{Q(S)})=\sum_{i=1}^n \fv{j_i}{Q(S)})$.
\end{enumerate} 
The above two observations imply that $c(S)=P-\defi{Q(S)}$ which proves the second part of the Lemma.
\end{proof}
By Lemma~\ref{le:fesible} an infeasible instance has a cut \cut{S} such that $c(S)<P$.  Lemma~\ref{le:deficiency} then implies that $\defi{Q(S)}>0$ which proves the following theorem.
\begin{theorem}
An instance of \dsoi\ is feasible iff no set of disjoint intervals has positive deficiency.
\end{theorem}  
{\bf Making an instance feasible}:
Given an infeasible instance of \dsoi, we would like to extend the intervals of the instance to make it feasible. We need some additional tools to do this and shall take this up in a later section. Let $F<P$ be the maximum $s$-$t$ flow in the network $G$ corresponding to this instance. We now show that an $s$-$t$ flow of value $P$ can be routed in $G$ by increasing capacities of edges incident to the sink such that the total increase in capacities is $P-F$.

By submodularity of the cut-function it follows that if \cut{S_1},\cut{S_2} are minimum $s$-$t$ cuts then \cut{S_1\cap S_2} is also a minimum $s$-$t$ cut. Hence a minimum $s$-$t$ cut in which the side containing the source is minimal is unique; let \cut{S} be this cut. Since the capacity of this cut is less than $P$, by Claim~\ref{cl:basic} it follows that $S\cap\set{v_0,v_1,\ldots,v_{D-1}}\neq\varnothing$.
\begin{clm}\label{cl:inc}
Increasing the capacity of {\em any} edge $(v_i,t), v_i\in S$ by 1 increases the $s$-$t$ max-flow in $G$ by 1. 
\end{clm}
\begin{proof}
For contradiction assume that increasing the capacity of edge $(v_i,t)$ does not increase the $s$-$t$ max-flow in $G$. Hence there is a minimum $s$-$t$ cut, \cut{X}, such that $v_i\not\in X$. Since $v_i\in S$, this means $S\not\subseteq X$ which implies that $S$ is not minimal. 
\end{proof}
Claim~\ref{cl:inc} gives us an algorithm for increasing capacities. At each step we find a minimum $s$-$t$ cut in which the side containing the source is minimal and increase the capacity of any edge in this cut which is also incident to the sink by 1. Since with every step, we increase the $s$-$t$ flow in $G$ by 1, the number of steps, and the total increase in edge capacities, equals $P-F$.

Claim~\ref{cl:inc} also implies that \cut{S} remains a minimum $s$-$t$ cut in $G$ even after we increase the capacity of edge $(v_i,t), v_i\in S$, by 1; however $S$ need not be minimal. Let \cut{S'} be the new $s$-$t$ minimum cut in which the side containing the source is minimal. The fact that \cut{S} is a minimum $s$-$t$ cut implies that $S'\subseteq S$. Thus with every step the $s$-side of the cut under consideration shrinks. This is an important property of this process and shall find use later.

\section{Linear Programming Relaxation}
\label{sec:lp}
We are now ready to give a linear programming relaxation for the problem of scheduling jobs on parallel machines so as to minimize total energy consumed. A solution to the problem is completely specified by the set of time intervals in which each machine is active; let $\mcI$ be this multiset. The energy consumed by this solution equals $\sum_{I\in\mcI} (\abs{I}+Q)$. Note that at most $m$ intervals in $\mcI$ can overlap at any point in time. Further, $\mcI$ forms a feasible solution if the corresponding instance of \dsoi\ is feasible.  

With every interval $I\subseteq [0,D]$ we associate a variable $x(I), 0\le x(I)\le m$ which indicates the number of times $I$ is picked in a solution. The objective is to minimize $\sum_I x(I)(\abs{I}+Q)$. We now list the constraints of this linear program.
\begin{enumerate}
\item Let $m_t=\sum_{I:[t,t+1]\subseteq I} x(I)$. Since at most $m$ intervals overlap at any time $t$ we get that for all $t, 0\le t\le D-1$, $m_t \le m$.
\item Let $f(i,t)$ be a variable denoting the flow in the edge $(u_i,v_t), 0\le i\le n, 0\le t\le D-1$ in the flow network $G$ corresponding to this instance. Then $0\le f(i,t)\le 1$.
\item The conservation constraint on vertex $v_t$ and the capacity constraint on edge $(v_t,t)$ together give: for all $t, 0\le t\le D-1$, $\sum_{i: [t,t+1]\subseteq[r_i,d_i]} f(i,t) \le m_t$.
\item Since $P$ units of flow have to be routed, all edges incident to the source are saturated. This together with the conservation constraint at vertex $u_i$ yields: for all $i$, $\sum_{t=r_i}^{d_i-1} f(i,t)=p_i$.
\item Consider an interval $[a,b]\subseteq[0,D]$. The total forced volume of all jobs with respect to $[a,b]$ equals $\sum_{i=1}^n \fv{j_i}{[a,b]}$. If this quantity equals $\alpha(b-a)$ then the number of intervals overlapping $[a,b]$ should be at least \ceil{\alpha}. This yields the constraint: for all $0\le a < b\le D$, $$\sum_{I:[a,b]\cap I\neq\varnothing} x_I \geq \bigg\lceil\frac{\sum_{i=1}^n \fv{j_i}{[a,b]}}{b-a}\bigg\rceil.$$
\end{enumerate}
Thus our linear program for scheduling on multiple machines to minimize energy is as follows.
\begin{lp}{minimize}{\sum_I x(I)(\abs{I}+Q)}
\cnstr{m_t}{=}{\sum_{I:[t,t+1]\in I} x(I)}{0\le t\le D-1}
\cnstr{m_t}{\ge}{\sum_{i: r_i\le t\le d_i-1} f(i,t)}{0\le t\le D-1}
\cnstr{p_i}{=}{\sum_{t=r_i}^{d_i-1} f(i,t)}{1\le i\le n}
\cnstr{\sum_{I:[a,b]\cap I\neq\varnothing} x_I}{\ge}{\big\lceil\sum_{i=1}^n \fv{j_i}{[a,b]}/(b-a)\big\rceil}{0\le a < b\le D}
\cnstr{f(i,t)}{\in}{[0,1]}{1\le i\le n, 0\le t\le D-1}
\cnstr{x(I),m_t}{\in}{[0,m]}{0\le t\le D-1, I\subseteq[0,D]}
\end{lp}
\section{Minimizing Energy on Parallel Machines}
\label{sec:multiple}
Our algorithm for the case of parallel machines is along the lines of the one for single machines. We begin by solving the linear program from Section~\ref{sec:lp} and let $x$ be the optimum fractional solution and $\opt$ the cost of this solution. Our algorithm will produce a solution of cost at most $2\opt+P$.

Let $\mcI=\set{I|x_I>0}$. After ensuring that no interval of $\mcI$ is strictly contained in another, we order the intervals by increasing start-times (breaking ties using end-times) and let $\prec$ be this order. As in Section~\ref{sec:single}, we construct $r$ integral solutions, $\mcC_i, 1\le i\le r$ and associate weights $w_i$ with solutions $\mcC_i$ such that this forms a convex decomposition of $x$. Note that $\mcC_i$ is no more a disjoint set of intervals as in the single machine case. However at most $m$ intervals of $\mcC_i$ could overlap at any point in time.

For the rest of this section we will consider one of the integral solutions in the convex decomposition and refer to it as $\mcC$. The arguments of this section will apply to all $r$ solutions. Note that $\mcC$ need not be a feasible instance of \dsoi\ and we will modify the intervals in $\mcC$ to make it a feasible solution. Let $I_1\prec I_2\prec\cdots\prec I_N$ be the intervals in $\mcC$.

\begin{lemma}
\label{overlap-bound}
Suppose $[a,b]\subseteq[0,D]$ overlaps $l$ intervals of $\mcC=\mcC_i$. Then $[a,b]$ overlaps at most $l+1$ intervals of $\mcC_k, k\neq i$. 
\end{lemma}

\begin{proof}
From our round-robin procedure for assigning intervals to solutions in the convex decomposition it follows that for any $1\le i\le N-1$, $\mcC_k$ contains exactly one interval $I$ between $I_i$ and $I_{i+1}$ i.e. $I_i\prec I\prec I_{i+1}$. Suppose $[a,b]$ overlaps intervals $I_j, I_{j+1},\ldots, I_{j+l-1}$ of $\mcC$. Then $[a,b]$ would definitely overlap the $l-1$ intervals of $\mcC_k$ between $I_j$ and $I_{j+l-1}$. In addition $[a,b]$ could possibly overlap the two intervals of $\mcC_k$ between $I_{j-1}$ and $I_j$ and between $I_{j+l-1}$ and $I_{j+l}$. Thus $[a,b]$ could overlap at most $l+1$ intervals of $\mcC_k$.
\end{proof}

{\bf Modifying intervals}: Let $s_j,e_j$ denote the start and end times of interval $I_j\in\mcC$. We consider the intervals in the order $\prec$ and modify them as follows:
\begin{quote}
    If $I_j, I_{j+1}$ overlap then replace $I_j$ with the interval $[s_j,e_{j+1}]$. Else create a copy of $I_{j+1}$ if it does not overlap $I_{j+m}$. 
\end{quote} 
For $j=0$ we add a copy of $I_1$ if it does not overlap $I_m$. The set of intervals formed through this modification continue to have the property that no interval is strictly contained in another although now we could have two copies of some intervals. Let $I'_1\prec I'_2\prec\cdots\prec I'_M$ be the new (multi)set of intervals which we denote by $\mcC'$. 
\begin{clm}\label{cl:adding}
The sets $\mcC$ and $\mcC'$ relate as:
\begin{enumerate}
\item The total length of the intervals in $\mcC'$ is at most twice the total length of intervals in $\mcC$.
\item The number of intervals in $\mcC'$ is at most twice the number of intervals in $\mcC$.
\item If $[a,b]\subseteq[0,D]$ overlaps $0<l<m$ intervals of $\mcC$ then it overlaps at least $l+1$ intervals of $\mcC'$.
\item At most $m$ intervals of $\mcC'$ overlap at any point in time.
\end{enumerate}
\end{clm}
\begin{proof}
The first 2 statements follow from our procedure for constructing $\mcC'$. The final statement of the claim follows from the fact that we add a copy of interval $I_{j+1}$ only if it does not overlap $I_{j+m}$.

To prove the third statement, suppose $[a,b]$ overlaps intervals $I_j, I_{j+1},\ldots, I_{j+l-1}$ of $\mcC$. Then $[a,b]$ would also overlap the corresponding intervals of $\mcC'$. Since $l<m$, $I_j$ does not overlap $I_{j+m-1}$.

Let $j>1$. If $I_{j-1}$ overlaps $I_j$ then $[a,b]$ would also overlap the interval in $\mcC'$ that replaced $I_{j-1}$. If $I_{j-1}$ does not overlap $I_j$ then $\mcC'$ would contain a copy of $I_j$ which $[a,b]$ would overlap. Finally if $j=1$ then we would have created a copy of $I_1$ in $\mcC'$ which $[a,b]$ would overlap. Thus $[a,b]$ would overlap at least $l+1$ intervals of $\mcC'$.
\end{proof}

{\bf Extending Intervals}: We will now extend intervals in $\mcC'$, without creating any new ones, to obtain a feasible instance of \dsoi. We begin by running the feasibility test of Section~\ref{feasible} on the instance whose supply-intervals are the intervals of $\mcC'$. Suppose the test fails and returns a set of intervals $Q=\set{Q_1,Q_2,\ldots Q_k}$ of maximum deficiency. Let $I'\in\mcC'$ be such that it overlaps $Q_i$ without containing $Q_i$ i.e. $Q_i\not\subseteq I'$. Then a time-slot in $Q_i$ can be used to extend $I'$ and doing this decreases the deficiency of $Q$ by 1. Recall that this also decreases the maximum deficiency of any set of intervals by 1. We modify the intervals in $\mcC'$ in this manner, always extending an interval of $\mcC'$ by a time-slot contained in one of the intervals comprising the set of intervals with maximum deficiency. We stop when it is not possible to extend an interval of $\mcC'$ in this manner and will now argue that the $\mcC'$ thus obtained is a feasible instance of \dsoi.

The intervals comprising $Q$ shrink during the above procedure and let \set{Q_1,Q_2,\ldots, Q_k} be the set of intervals with maximum deficiency when we stop. Let $Q_i=[a_i,b_i]$ and $c_i$ be the number of intervals of $\mcC'$ which overlap $Q_i$.
\begin{lemma}\label{le:key}
The number of intervals in $\mcC_j, 1\le j\le r$ which overlap $Q_i$ is at most $c_i$.
\end{lemma}
\begin{proof}
If $c_i=0$ then no interval in $\mcC'$ overlaps $Q_i=[a_i,b_i]$. Since in going from $\mcC$ to $\mcC'$ we have only extended intervals or introduced new intervals, this implies that no interval in $\mcC$ overlaps $[a_i,b_i]$. By our convex decomposition procedure this implies that $\sum_{I:I\cap[a_i,b_i]\neq\varnothing} x_I < 1$. Since $x$ is a feasible solution to the linear program (Section~\ref{sec:lp}) we conclude that $\sum_{k=1}^n \fv{j_k}{[a_i,b_i]}=0$.

Since $Q$ is a minimal set of intervals with maximum deficiency $\defi{Q\setminus Q_i} <\defi{Q}$. Since no interval of $\mcC'$ overlaps $Q_i$ this implies $\sum_{k=1}^n \fv{j_k}{Q\setminus Q_i} < \sum_{k=1}^n \fv{j_k}{Q}$. Hence there exists a job $j_k$ such that $\fv{j_k}{Q\setminus Q_i}<\fv{j_k}{Q}$. This implies that $[r_k,d_k]\cap Q_i\neq\varnothing$. Further $[r_k,d_k]\not\subseteq Q_i$ as that would imply $\fv{j_k}{Q_i}>0$. Hence either $r_k< a_i<d_k$ or $r_k<b_i<d_k$; note that both conditions could also be true. 

If $r_k< a_i<d_k$ then expanding $Q_i$ to $[a_i-1,b_i]$ would increase \fv{j_k}{Q} by 1. Since $Q$ is a set of intervals with maximum deficiency, some interval of $\mcC'$ must include the time-slot $[a_i-1,a_i]$. Similarly, if $r_k<b_i<d_k$ then by expanding $Q_i$ to $[a_i,b_i+1]$ we conclude that an interval of $\mcC'$ contains $[b_i,b_i+1]$. In either case, we have an interval of $\mcC'$ overlapping $[a_i,b_i]$ which implies $c_i>0$. 

By the third statement of Claim~\ref{cl:adding} the number of intervals in $\mcC$ overlapping $Q_i$ is at most $c_i-1$. Then by Lemma~\ref{overlap-bound} the number of intervals in $\mcC_j$ overlapping $Q_i$ is at most $c_i$ and this proves the lemma. 
\end{proof}

Consider $x$, the optimum solution to the LP. For all $t$ such that $[t,t+1]\subseteq Q_i$ we have, $$m_t=\sum_{I:[t,t+1]\subseteq I} x_I=\sum_{j=1}^r w_j\abs{\set{I\in\mcC_j, [t,t+1]\subseteq I}}\le \sum_{j=1}^r w_jc_i = c_i,$$ where the inequality follows from Lemma~\ref{le:key}.

Consider the cut \cut{S} where $S=\set{s}\cup\set{v_t|[t,t+1]\in Q}\cup\set{u_i| \fv{j_i}{Q}>0}$. The capacity of this cut is 
$$\sum_{t:[t,t+1]\subseteq Q} m_t + \sum_{i=1}^n (p_i-\fv{j_i}{Q}) \le P-\left(\sum_{i=1}^n \fv{j_i}{Q}-\sum_{i=1}^k c_i\abs{Q_i}\right)= P-\defi{Q},$$ where $\defi{Q}$ is the deficiency of the set of intervals $Q$ for an instance of \dsoi\ defined by intervals of $\mcC'$. If $\defi{Q}>0$ then $c(S)<P$ which contradicts the feasibility of $x$. Thus $\defi{Q}=0$ and so the intervals of $\mcC'$ form a feasible solution.

By Claim~\ref{cl:adding}, the total energy consumption of intervals in $\mcC'$ is at most twice that of the intervals in $\mcC$. Our procedure for extending intervals in $\mcC'$ increases their total length, and hence the total energy, by at most $P$. Hence the solution of minimum cost among $\mcC'_j, 1\le j\le r$ has cost at most $2\opt+P$ where $\opt$ is the cost of the optimum fractional solution.
\begin{theorem}
Given $n$ jobs with release dates, processing times and deadlines in $[0,D]$, there is an algorithm with running time polynomial in $n,D$ which schedules these jobs on $m$ machines such that the total energy consumption is at most $2\opt +P$ where $P$ is the sum of processing times. 
\end{theorem}

\section{Conclusions}
The two algorithms with running times polynomial in $n$ and $D$ can be converted to polynomial time algorithms by limiting the number of intervals we consider in the linear program and by suitably modifying our procedure for extending the intervals in the integral solutions of the convex decomposition (see Appendix). We believe that our approach of formulating this problem of minimizing energy as a linear program and the tools we develop in this paper for rounding the fractional solutions, hold much promise and can be applied to more general machine models and power management techniques.  
\bibliographystyle{plain}
\bibliography{references}
\appendix
\section*{Appendix}
\section{From Pseudopolynomial to Polynomial Time}
\label{sec:running}
In this section we prove that it is sufficient to limit ourselves to intervals with start and endpoints from a set $W$ of polynomially many time-slots in $[0,D]$, with the loss of a small factor in the approximation ratio.

\begin{definition}
Let $T:=\cup_i\{r_i,d_i\}$ and $W:=T\cup \{w| w\in [0,D], \exists t\in T, k\in\mathbb{N} \text{ with } |t-w|=\lceil (1+\epsilon)^k\rceil\}\cup\{0,D\}$.
\end{definition}

\begin{clm}
    \label{clm:number-points}
    $|W|$ is polynomial in the input size, and therefore so is the number of possible intervals that start and end at time-slots of $W$.
\end{clm}
\begin{proof}
    Consider some $t\in T$.  We argue about the number of distinct $w\in[0,D]$, so that $|t-w|=\lceil (1+\epsilon)^k\rceil$ for some $k\in\mathbb{N}$. Since $t,w\in[0,D]$, we have that $|t-w|\le D+1$ and therefore $k=O(\log D)$. In turn $|W|=|T|\cdot k = O(n\log D)$,
    and the number of possible intervals starting and ending at $W$ is $|W|^2= O(n^2\log^{2} D)$.
\end{proof}

\begin{lemma}
\label{lem:running-time}
    Considering only intervals that start and end at time-slots of W, does not  increase the cost of being in the active state by more than a factor of $(1+\epsilon)$.
\end{lemma}
\begin{proof}
    Consider an optimal solution \opt. We will transform \opt\ to a solution that satisfies the lemma property while increasing its active cost by at most an $(1+\epsilon)$ factor. We can associate intervals in \opt\ with processors as follows. Recall that the intervals are ordered by their start-time and ties are broken by end-times. Go through the intervals in this order and associate each interval to the smallest-index processor so that it does not overlap with any other interval already there. We will use the following claim to prove the lemma:
    \begin{clm}
        \label{clm:contains-point}
        Assuming that $I\cap T\neq\varnothing$ holds for any interval $I\in\opt $  is without loss of generality.
    \end{clm}
        Consider some interval $I\in\opt$, associated with a processor $m_I$ and let $t\in I$ be a time slot of $T$ whose existence is guaranteed by Claim~\ref{clm:contains-point}.
    
     We will expand $I$ towards the left and the right respectively until we hit either some time slot in $W$ or we hit another interval associated with this processor. In the second case we merge the two intervals. We repeat this for every interval, and the process will terminate since in each step we either "snap" one of the endpoints to a point in $W$ or reduce the number of intervals by one. Note that eventually all interval endpoints will be slots in $W$ ($W$ includes $0$ and $D$).
     
     The total increase in length of an interval $I$ is at most $(1+\epsilon)\cdot |I|$, because we expand towards the left by at most a factor $(1+\epsilon)\cdot |t-s_I|$, and similarly towards the right by at most a factor $(1+\epsilon)\cdot |e_I-t|$. This is because by construction there are points in $W$ at every $(1+\epsilon)$ multiple distance away from $t$, and we never expand more than that.
        
        We conclude the proof of the lemma by proving Claim~\ref{clm:contains-point}.
    \begin{proof}[Proof of Claim~\ref{clm:contains-point}]
    Assume for the sake of contradiction that there exists an $I\in\opt$ such that $I\cap T=\varnothing$. Then we  move $I$ towards an adjacent point $t\in T$. Without loss of generality assume that we move $I$ leftwards. So consider moving $I$ leftwards one slot at a time. We break up this moving of $I$ one slot leftwards into consecutively moving all units of $I$ one slot leftwards: We first move the leftmost unit, then the next one etc. The following could potentially happen:
    \begin{itemize}
        \item Interval $I$ reaches $t$. In this case $I\cap T \neq\varnothing$ and we stop.
        \item Interval $I$ meets the endpoint $e_{I}$ of some other interval $I'$ on the same or a different processor. This cannot happen since it would contradict the optimality of \opt. The reason is that one can either merge $I$ with $I'$, or use part of $I$ to close the gap following $I'$ on its processor. Either requires one wake-up operation less but has otherwise identical costs to \opt.
        \item We are not able to move some unit of $I$ one more slot leftwards without producing an infeasible schedule. Since there is still no point in $T$ intersecting $I$ this must be because some job $j$ running in this unit of $I$ would run in parallel to itself if we move the interval one more slot leftwards. Let $\ell$ be the slot on which $j$ runs in $I$, and assume that it runs in some slot $\ell-1$ on some other processor. If there is some interval $I'$ on one of the other processors ending at slot $\ell-1$, we simply move the unit of $j$ to that processor continue shifting the remaining slots of $I$ to the left. Thus we may assume that slot $\ell-1$ contains strictly less jobs than slot $\ell$. By the pigeon hole principle there exists some job that we can swap with $j$ in slot $\ell$ so that we can move one more unit of $I$ one slot leftwards. 
    \end{itemize}
    Since in each step we move one unit of $I$ one slot leftwards, the process will eventually terminate with $I\cap T\neq \varnothing$. Note that the process does not increase the number of intervals, nor the sum of interval lengths (although it may change individual interval lengths), and therefore does not affect the cost of the solution. 
    \end{proof}
\end{proof}

\textbf{Modifying the Flow Network and Linear Program.}  We first show how to modify the network for checking the feasibility of \dsoi. Let $W=\{a_{0},a_{1},\ldots, a_{k}\}$, with $a_{0} < a_{1} < \ldots < a_{k}$. The consecutive points in $W$ partition $[0,D]$ into $k$ time intervals, ie. $I_{W}=\{[a_{0},a_{1}],\ldots,[a_{k-1}a_{k}]\}$. We refer to the interval $[a_{k-1},a_{k}]$ as the $k^{th}$ time slot. We next discuss how to adapt the maximum flow formulation. Firstly, instead of nodes $v_{t}$ for each time $t, 1 \leq t \leq D-1$, we now have a node $v_{t}, 1\leq t \leq k$ for each time slot in $I_{W}$. The capacity of edge $(u_{i},v_{t})$ is the length of interval $[a_{t-1},a_{t}]$. Let $n_t$ be the number of intervals crossing time slot $t$. The capacity of edge $(v_{i},t)$ is $m_t$, where $m_t$ is defined as  the product of $n_t$ and the length of time slot $t$. Note that size of the network after doing the above modification is $O(n|W|)$. As in Lemma \ref{le:fesible}, we can again argue that the given instance is feasible iff $P$ units of flow can be routed in the network. If the instance is feasible, then $P$ units of flow can clearly be routed. Suppose $P$ units of flow can be routed in the network. Fix a time slot $t$. We have to schedule $f(i,t)$ units of job $i$ in the $t^{th}$ time slot such that $f(i,t) \leq |a_{t}-a_{t-1}|$ and $\sum_{i}f(i,t) \leq m_{t}=n_{t}|a_{t}-a_{t-1}|$. Consider a schedule of all jobs (active in time slot $t$) on a single machine such that job $i$ is processed for $f(i,t)$ units, every job is processed contiguously and there is no gap in the schedule. The machine runs continuously in $[0,\sum_{i}f(i,t)]$. We replicate the schedule of this machine in time  $[(i-1)|a_{t}-a_{t-1}|,i|a_{t}-a_{t-1}|]$ on the $i^{th}$ interval crossing time slot $t$. No job is processed in two intervals at the same time as no job has length more than $|a_{t}-a_{t-1}|$. We modify appropriate constraints in the Linear Program to reflect changes made in the network.\\

\textbf{Modifying the Rounding Procedure.}
We now argue that the rounding procedure of Section \ref{sec:multiple} can be carried out in polynomial time. The algorithm works in iterations. In each iteration, the rounding procedure finds a minimal set of intervals of maximum deficiency and increases the length of an interval in this set by 1. This results in reduction of maximum deficiency by 1 and there can be at most $P$ such iterations. We make the following minor modification to this algorithm. If we decide to extend an interval $I$ of some solution $\mcC_j$ in an iteration, we extend it by $\delta$, where $\delta$ is the maximum number such that extending $I$ by $\delta$ also reduces the maximum deficiency of this solution by $\delta$. We can find such a $\delta$ by binary search. Recall that minimal maximum deficiency set shrinks after every iteration. Suppose $Q$ is the minimal maximum deficiency set after $I$ was extended by $\delta$. Since $I$ was not extended any further (in a previous iteration), either it does not overlap with $Q$ or none of the endpoints of $I$ are inside $Q$. In either case, this interval will never be extended in any further iteration. Hence, the total number of iterations is bounded by the maximum number of intervals in a solution, which is $O(m|W|)$. Since total number of solutions is at most the number of possible intervals, the total number of iterations required for constructing all the solutions is at most $O(m|W|^{3})$. Also, the total length of intervals added to a solution is equal to the maximum deficiency, which is at most $P$ and hence the rounding procedure does not further affect the approximation guarantee of the algorithm. After extending the intervals, each solution has a maximum deficiency of zero and hence feasible (by discussion in the last section).

\section{Integrality Gap Example}
Consider an instance on a single machine with 5 jobs, $j_{1},\ldots,j_{5}$ (see Figure \ref{fig:bad-example}). Let $r_{1}=0,d_{1}=1,r_{2}=1,d_{2}=7,r_{3}=2,d_{3}=4,r_{4}=4,d_{4}=6,r_{5}=7,d_{5}=8$. All jobs have unit processing time and the wake up cost of the machine is $1$. Since the wake up cost is $1$ we may assume that the machine transitions to the sleep state whenever it is idle, in other words there exists an optimal (integral) solution with no active but idle periods. We claim that the aforementioned instance requires at least three contiguous active time intervals: First note that $j_{1}$ and $j_{5}$ have to be done in time slots $[0,1]$ and $[7,8]$ respectively and since there are only three units of work to be done in $[1,7]$, $j_1$ and $j_5$ must be processed in two different intervals. Let $I_{1}$ and $I_{2}$ be these respective intervals. If $j_{2}$ is processed in $I_{1}$, then $j_{4}$ cannot be processed in $I_{1}$ or $I_{2}$. Similarly, if $j_{2}$ is processed in $I_{2}$, then $j_{3}$ cannot be processed in $I_{1}$ or $I_{2}$. Hence, optimal solution must incur wake up energy of at least $3$ and total energy of optimal solution is at least $8$.

We now show a fractional solution with value strictly smaller than 8. Let $I_{1}=[0,1],I_{2}=[0,3],I_{3}=[4,6],I_{4}=[5,8],I_{5}=[7,8]$ (see Figure \ref{fig:bad-example}). Consider a fractional solution with $x_{I_{1}}=x_{I_{2}}=x_{I_{3}}=x_{I_{4}}=x_{I_{5}}=1/2$. It can easily be verified that this is a feasible fractional solution with energy 15/2. Hence, integrality gap of the LP is at least $16/15$. 
\end{document}